\documentclass[american,aps,pra,reprint,floatfix,nofootinbib,superscriptaddress]{revtex4-1}
\usepackage[unicode=true,pdfusetitle, bookmarks=true,bookmarksnumbered=false,bookmarksopen=false, breaklinks=false,pdfborder={0 0 0},backref=false,colorlinks=false]{hyperref}
\hypersetup{colorlinks,linkcolor=myurlcolor,citecolor=myurlcolor,urlcolor=myurlcolor}
\usepackage{graphics,epstopdf,graphicx,amsthm,amsmath,amssymb,mathptmx,braket,colortbl,color,bm,framed,mathrsfs}
\usepackage[T1]{fontenc}
\usepackage[up]{subfigure}
\usepackage{tikz}

\definecolor{myurlcolor}{rgb}{0,0,0.9}

\newcommand{\proj}[1]{| #1\rangle\!\langle #1 |}

\newcommand{\iinner}[2]{\langle #1 | #2\rangle}

\DeclareMathOperator{\trace}{Tr}
\newcommand{\Ptr}[2]{\trace_{#1}\Pa{#2}}
\newcommand{\Tr}[1]{\Ptr{}{#1}}

\newcommand{\Pa}[1]{\left[#1\right]}

\newcommand{\norm}[1]{\left\lVert #1 \right\rVert}

\theoremstyle{plain}
\newtheorem{thm}{Theorem}
\newtheorem{lem}[thm]{Lemma}
\newtheorem{prop}[thm]{Proposition}
\newtheorem{cor}[thm]{Corollary}

\newcommand*{\myproofname}{Proof}

\def\ot{\otimes}
\def\complex{\mathbb{C}}

\def\sgn{\mathrm{sgn}}

\DeclareMathAlphabet{\mathcal}{OMS}{cmsy}{m}{n}

\makeatother

\begin{document}

  \author{Kaifeng Bu}
 \email{kfbu@fas.harvard.edu}
 \affiliation{School of Mathematical Sciences, Zhejiang University, Hangzhou, Zhejiang 310027, China}
\affiliation{Department of Physics, Harvard University, Cambridge, Massachusetts 02138, USA}

  \author{Dax Enshan Koh}
  \email{daxkoh@mit.edu}
\affiliation{Department of Mathematics, Massachusetts Institute of Technology, Cambridge, Massachusetts 02139, USA}

\title{Efficient classical simulation of Clifford circuits with nonstabilizer input states}

\begin{abstract}
We investigate the problem of evaluating the output probabilities of Clifford circuits with nonstabilizer product input states. First, we consider the case when the input state is mixed, and give an efficient classical algorithm to approximate the output probabilities, with respect to the $l_1$ norm, of a large fraction of Clifford circuits. The running time of our algorithm decreases as the inputs become more mixed. Second, we consider the case when the input state is a pure nonstabilizer product state, and show that a similar efficient algorithm exists to approximate the output probabilities, when a suitable restriction is placed on the number of qubits measured. This restriction depends on a magic monotone that we call the Pauli rank.
We apply our results to give an efficient output probability approximation algorithm for some restricted quantum computation models, such as Clifford circuits with solely magic state inputs (CM), 
Pauli-based computation (PBC) and instantaneous quantum polynomial time (IQP) circuits.

\end{abstract}

\maketitle

\section{Introduction}

One of the main motivations behind the field of quantum computation is the expectation that quantum computers can solve certain problems much faster than classical computers. This expectation has been driven by the discovery of quantum algorithms which can solve certain problems believed to be intractable on a classical computer. A famous example of such a quantum algorithm is due to Shor, whose eponymous algorithm can solve the factoring problem exponentially faster than the best classical algorithms we know today \cite{shor1994algorithms, shor1999polynomial}.

With the advent of noisy intermediate-scale quantum (NISQ) devices \cite{preskill2018quantum}, an important near-term milestone in the field is to demonstrate that quantum computers are capable of performing computational tasks that classical computers cannot, a goal known as 
\textit{quantum supremacy} \cite{preskill2012quantum, harrow2017quantum}. Several restricted models of quantum computation have been proposed as candidates for demonstrating quantum supremacy. These include
boson sampling \cite{aaronson2011computational}, the one clean qubit model (DQC1) \cite{knill1998power, fujii2018impossibility}, instantaneous quantum polynomial-time (IQP) circuits \cite{bremner2010classical}, Hadamard-classical circuits with one qubit (HC1Q) \cite{Morimae2018merlinarthur}, Clifford circuits with magic initial states and nonadaptive measurements \cite{jozsa2014classical, koh2015further, yoganathan2018quantum}, the random circuit sampling model \cite{boixo2018characterizing,bouland2018onthecomplexity}, and conjugated Clifford circuits (CCC) \cite{bouland2018complexity}.
These models are potentially good candidates for quantum supremacy because they can solve sampling problems that are conjectured to be intractable for classical computers, and are conceivably easier to implement in experimental settings.

In contrast to the above models, quantum circuits with Clifford gates and stabilizer input states are not a candidate for quantum supremacy, because they can be efficiently simulated on a classical computer using the Gottesman-Knill simulation algorithm \cite{gottesman1997heisenberg}. The Gottesman-Knill algorithm, however, breaks down and efficient classical simulability can be proved to be impossible (under plausible assumptions) when Clifford circuits are modified in various ways, under various notions of simulation \cite{jozsa2014classical, koh2015further, bouland2018complexity, yoganathan2018quantum}. For example, it can be proved under plausible complexity assumptions that no efficient classical sampling algorithm exists that can sample from the output distributions of  Clifford circuits with general product state inputs when the number of measurements made is of order $O(n)$ \cite{jozsa2014classical}.

In this paper, we present two new efficient classical algorithms for approximately evaluating the output probabilities of Clifford circuits with nonstabilizer inputs. 
Our first algorithm shows that the output distribution of Clifford circuits with mixed product states can be efficiently approximated, with respect to the $l_1$ norm, for a large fraction of Clifford circuits. This algorithm explicitly reveals the role of mixedness of the input states in affecting the running time of the simulation, which decreases as the inputs become more mixed.

Our second algorithm shows that such an efficient approximation algorithm still exists in the case where the inputs are pure nonstabilizer states, as long as we impose a suitable restriction on the number of measured qubits. This restriction depends on a magic monotone called the Pauli rank that we introduce in this paper. 
This algorithm also explicitly links the simulation time to the amount of magic in the input states, and implies that for Clifford circuits with magic input states, it is possible in certain cases to achieve an efficient classical approximation
of the output probability even when $O(n)$ qubits are measured. This is in contrast to the hardness result in \cite{jozsa2014classical}, which shows that sampling from those output probabilities is hard. Finally, we apply our results to give an efficient approximation algorithm for some restricted quantum computation models, like Clifford circuits with solely magic state inputs (CM), 
Pauli-based computation (PBC) and instantaneous quantum polynomial time (IQP) circuits.

\section{Main results}

Let $P^n$ be the set of all Hermitian Pauli operators on $n$ qubits, i.e., operators that can be written as
the $n$-fold tensor product of the single-qubit Pauli operators $\set{I, X, Y, Z}$ with sign $\pm 1$. 
The Clifford unitaries on $n$ qubits are the unitaries that maps Pauli operators to Pauli operators, that is,
$\mathcal{C}l_n=\set{U\in U(2^n): UPU^\dag\in P^n, \forall P\in P^n}$.  Stabilizer states are pure states of the form $U\ket{0}^{\ot n}$ \cite{Scott2004}, where $U$ is some Clifford unitary.

Here, we consider Clifford circuits with product input states $\proj{0}^{\ot n}\ot^m_{i=1} \rho_i$, and measurements on $k$ qubits.
If either $m$ or $k$ is $O(\log n)$, the output probabilities  can be efficiently simulated classically by the Gottesman-Knill theorem 
\cite{gottesman1997heisenberg,jozsa2014classical}. 
However, if both $m$ and $k$ are greater than $O(\log n)$, we show that the output probability of such circuits can still be 
approximated efficiently with respect to the $l_1$ norm for a large fraction of Clifford circuits.

 \begin{figure}[!h]
  \center{\includegraphics[width=5cm]  {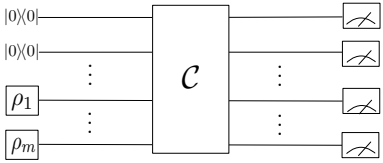}}     
  \caption{A circuit diagram of Clifford circuits with product state inputs, which could be either pure or mixed.}
 \end{figure}

\subsection{Mixed input states}
We first
consider the case where
all $\rho_i$ are mixed states and give an efficient classical algorithm to approximate the output 
probabilities.

\begin{thm}\label{thm:main}
Given a Clifford  circuit $\mathcal{C}$ on $n+m$ qubits with input state
$\proj{0}^{\ot n}\ot^{m}_{i=1} \rho_i$  and measurement on each qubit in the computational basis,
 there exists a classical algorithm 
to approximate the output probabilities of the circuit up to $l_1$ norm $\delta$ in time 
$(n+m)^{O(1)}m^{O(\log(\sqrt{\alpha}/\delta)/\lambda)}$ for at least $1-\frac{2}{\alpha}$ fraction of circuits $\mathcal{C}$, 
where $\lambda=\min\set{\lambda_i}_i$, with $\lambda_i=1-\sqrt{2\Tr{\rho^2_i}-1}$, is a measure of the mixedness of the input state $\rho_i$.
\end{thm}

The proof of the Theorem is presented in Appendix \ref{sec:mix}. 
The theorem shows that the efficiency of the classical simulation increases with the mixedness 
of the input states.

Next, we show that the result in Theorem \ref{thm:main} can be easily generalized to quantum circuits $\mathcal{C}$ which are slightly beyond Clifford circuits.
To this end, we consider the Clifford hierarchy, a class of operations introduced by Gottesman and Chuang \cite{gottesman1999demonstrating} that has important applications in fault-tolerant quantum computation and teleportation-based state injection.
 Let 
$\mathcal{C}l^{(3)}_{n}$ be the third level of the Clifford Hierarchy, i.e., 
$\mathcal{C}l^{(3)}_{n}=\set{U\in U(2^n): UPU^\dag\in \mathcal{C}l_n, \forall P\in P^n }$.
There are several important gates in the third level of Clifford Hierarchy, such as the $\pi/8$ gate (which we denote $T$) and the $CCZ$ gate  \cite{zeng2008semi}. 
(Note that the set $\mathcal{C}l^{(3)}_{n}$ is not closed under multiplication. For example, 
 $TH, T\in \mathcal{C}l^{(3)}_{n}$, but $THT\notin  \mathcal{C}l^{(3)}_{n}$.)
The following corollary shows that adding gates in $\mathcal{C}l^{(3)}$ to the circuits in Theorem \ref{thm:main} does not change (up to polynomial overhead) the efficiency of the classical simulation. 

\begin{cor}
Let $\mathcal{C}=\mathcal{C}_1\circ V$ be a quantum circuit with input states 
$\proj{0}^{\ot n}\ot^m_{i=1} \rho_i$, where the gates in the circuit $\mathcal{C}_1$ are taken from the 
set of Clifford gates on $n+m$ qubits $\mathcal{C}l_{n+m}$ and $V$ is taken from the third level of Clifford hierarchy
 $\mathcal{C}l^{(3)}_{m}$ acting on $n+1,...,n+m$-th qubits. Assume that each each qubit is measured in the computational basis.
Then, Theorem \ref{thm:main} stil holds if we replace $\mathcal C$ in Theorem \ref{thm:main} with $\mathcal C$ defined above.
\end{cor}
The key property we use here is that the gates in the third level of the Clifford Hierarchy
map Pauli operators to Clifford unitaries, which makes the proof of Theorem \ref{thm:main} still hold. (See a discussion of this in Appendix \ref{sec:mix}. )
Although $\mathcal{C}l^{(3)}_n$ is not a group, the diagonal gates in 
 $\mathcal{C}l^{(3)}_n$, denoted as $\mathcal{C}l^{(3)}_{n,d}$, forms  a group \cite{zeng2008semi,Cui2017}. Since the $T$ gate and $CCZ$ gate both belong to 
 $\mathcal{C}l^{(3)}_{n,d}$, the result in Theorem \ref{thm:main} still holds for the quantum circuits $\mathcal{C}=\mathcal{C}_1\circ\mathcal{C}_2$ where
 gates in $\mathcal{C}_1$ and $\mathcal{C}_2$ are chosen from $\mathcal{C}_{n+m}$ and $ \mathcal{C}l^{(3)}_{m,d}$ respectively.

Since noise is inevitable in real physical experiments, it is important to consider the effects of noise in quantum computation.  Recently, it has been demonstrated that if there 
is some noise on the random quantum gates \cite{Gao2018} or measurements of IQP circuits \cite{Bremner2017quantum}, then 
there exists an efficient classical simulation
of  the output distribution of 
quantum circuits. In the rest of this subsection, we apply our results to two important subuniversal quantum circuits with  noisy  input states and give an efficient
classical approximation  algorithm for the output probabilities of the corresponding quantum circuits.

\textit{\textbf{Example 1}}---First, we consider Clifford circuits with magic input states.
 \begin{figure}[!h]
  \center{\includegraphics[width=5cm]  {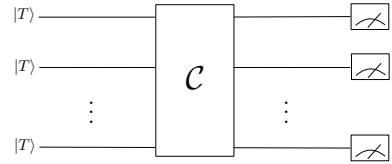}}  
  \caption{An example of a CM circuit}
 \end{figure}
It is well known that the Clifford + $T$ gate set is universal for quantum computation. By magic state injection, circuits with this gate set can be efficiently simulated by Clifford circuits with magic state
$\ket{T}$ inputs, where $\ket{T}=\frac{1}{\sqrt{2}}(\ket{0}+e^{i\pi/4}\ket{1})$.
It has been shown that 
 $\mathsf{postCM}=\mathsf{postBQP}$ \cite{yoganathan2018quantum}, and thus output probabilities are $\#\mathsf P$-hard approximate up to some constant relative error \cite{Kuperberg2015,FujiiNJP2017,Hangleiter2018quant}. 
However, if there is some independent depolarizing error acting on each input magic state, e.g., the input state on each register is
$(1-\epsilon)\proj{T}+\epsilon\frac{I}{2}$, then Theorem \ref{thm:main} implies directly that 
 there exists a classical algorithm 
to approximate the output probability up to $l_1$ norm $\delta$ in time 
$n^{O(\log(1/\delta)/\epsilon)}$ for a large fraction of the CM circuits with noisy inputs.

\textit{\textbf{Example 2}}---IQP circuits have a simple structure with input states $\ket{0}^{\ot n}$ and gates of the form $H^{\ot n} D H^{\ot n}$, 
where the diagonal gates in $D$  are chosen from the gate set $\set{Z, S, T, CZ}$.
 \begin{figure}[!h]
  \center{ \includegraphics[width=5cm]  {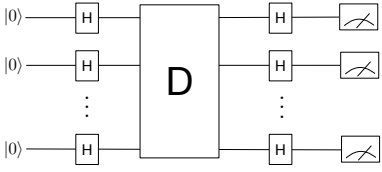}} 
  \caption{An example of an IQP circuit.}
 \end{figure}
 It has been shown that $\mathsf{postIQP}= \mathsf{postBQP}$ \cite{bremner2010classical} and thus, the output probabilities  are $\#\mathsf P$-hard to
approximate up to some constant relative error \cite{Kuperberg2015,FujiiNJP2017,Hangleiter2018quant}. 
Also, if there is some depolarizing noise acting on each input state $\ket{0}$, i.e., each input state is a mixed state  $(1-\epsilon)\proj{0}+\epsilon\frac{I}{2}$, then Theorem \ref{thm:main} implies that 
 there exists a classical algorithm 
to approximate the output probability up to $l_1$ norm $\delta$ in time 
$n^{O(\log(1/\delta)/\epsilon)}$ for a large fraction  of such IQP circuits. (The proof is presented in Appendix \ref{sec:IQP} in detail, which depends on the 
output distribution of IQP circuits in Appendix \ref{sec:dis_IQP}. )

\subsection{Pure nonstabilizer input states}
As we can see, the running time in Theorem \ref{thm:main} blows up if the input state
$\rho_i$ is pure. Here, we consider the case where all $\rho_i$ are pure nonstabilizer states, that is
Clifford gates with the input state $\ket{0}^{\ot n}\ot^m_{i=1}\ket{\psi_i}$.

For pure states $\ket{\psi}$,  the  stabilizer fidelity \cite{Bravyi2018} is defined as follows
\begin{eqnarray}
F(\psi)=\max_{\ket{\phi}}|\iinner{\phi}{\psi}|^2,
\end{eqnarray}
where the maximization is taken over all stabilizer states. 
Here, we define
\begin{eqnarray}\label{eq:mu}
\mu(\psi):=2(1-F(\psi)).
\end{eqnarray}
It is easy to see that $\mu(\psi)=0$ iff $\ket{\psi}$ is a stabilizer state. 
Thus, $\mu$ quantifies the distance between a given state to the set of stabilizer states. 
Since each $\ket{\psi_i}$ is not a stabilizer state, it follows that  $\mu(\psi_i)>0$.

Next, let us introduce the Pauli rank for pure single qubit states $\ket{\psi}$. First, we write a pure state
$\ket{\psi}$ in terms of its Bloch sphere representation 
$\proj{\psi}=\frac{1}{2}\sum_{s,t\in\set{0,1}}\psi_{st}X^sZ^t$, where
$\psi_{00}=1$ and $|\psi_{01}|^2+|\psi_{10}|^2+|\psi_{11}|^2=1$. 
We define the Pauli rank $\chi(\psi)$ to be the number of nonzero
coefficients $\psi_{st}$. 
By the definition of Pauli rank, it is easy to see that $2\leq \chi(\psi)\leq 4$, and
that $\ket{\psi}$ is a stabilizer state iff $\chi(\psi)=2$. Since each  input state $\ket{\psi_i}$
is a nonstabilizer state, it follows that $\chi(\psi_i)=3~ \text{or} ~ 4$. For example, for the magic state $\ket{T}$, 
the corresponding Pauli rank $\chi=3$.
For $n$-qubit systems, the Pauli rank serves as a good candidate for a magic monotone as it is easier to compute than other
magic monotones which require a minimization over all stabilizer states \cite{BravyiPRL2016,HowardPRL2017,Veitch2014}.
(See a discussion of Pauli rank for $n$-qubit systems in Appendix \ref{sec:non_stab}.)

\begin{thm}\label{thm:main2}
Given a Clifford circuit $\mathcal{C}$ on $n+m$ qubits  with input state 
$\ket{0}^{\ot n}\ot^m_{i=1} \ket{\psi_i}$   and measurements on $k$  qubits in the computational basis with 
$k\leq n+m-\sum^m_{i=1}\log_2(\chi(\psi_i)/2)$ and $\chi(\psi_i)$ being the Pauli rank of 
$\psi_i$,
 there exists a classical algorithm 
to approximate the output probability up to $l_1$ norm $\delta$ in time 
$(n+m)^{O(1)}m^{O(\log(\sqrt{\alpha}/\delta)/\mu)}$  for at least a $1-\frac{2}{\alpha}$ fraction of  Clifford circuits $\mathcal{C}$, 
where $\mu:=\min_i \mu(\psi_i)$ and $\mu(\psi_i)$ is defined as \eqref{eq:mu}. 
\end{thm}
The proof is presented in Appendix \ref{sec:non_stab}.
The maximal number of allowed measured qubits in this algorithm decreases with the amount of the magic in the input states, which is quantified by the Pauli rank. 
Curiously, the running time of this algorithm 
scales with the decrease in the amount of magic of the input states quantified by fidelity. This is contrary to the intuition that quantum circuits with more magic are harder to simulate.
Similarly, if the quantum circuits are slightly beyond the Clifford circuits, for example, 
$\mathcal{C}=\mathcal{C}_1\circ V$  where the gates in $\mathcal{C}_1$
are Clifford gates in $\mathcal{C}l_{n+m}$ and $V$ is some unitary gate in the third level of the Clifford Hierarchy $\mathcal{C}l^{(3)}_m$, then the result in 
Theorem \ref{thm:main2} still holds. 

Combining Theorem \ref{thm:main} and \ref{thm:main2}, we have 
the following corollary for any product input state:
\begin{cor}
Let $\mathcal{C}$ be a Clifford  circuit  on $n+m_1+m_2$ qubits with input states 
$\proj{0}^{\ot n}\ot^{m_1}_{i=1}\rho_i\ot^{m_2}_{j=1} \proj{\psi_j}$, where each $\rho_i$ is a mixed state, and each $\ket{\psi_j}$ is a pure
nonstabilizer state. Assume that measurements are performed on  $k$  qubits in the computational basis,  where
$k\leq n+m_1+m_2-\sum^{m_2}_{j=1}\log_2(\chi(\psi_i)/2)$ and $\chi(\psi_i)$ is the Pauli rank of 
$\psi_i$. Then,
 there exists a classical algorithm 
to approximate the output probability with respect to the $l_1$ norm $\delta$ in time 
$(n+m_1+m_2)^{O(1)}(m_1+m_2)^{O(\log(\sqrt{\alpha}/\delta)/\epsilon)}$  for at least $1-\frac{2}{\alpha}$ fraction of  Clifford circuits $\mathcal{C}$, 
where $\epsilon=\min\set{\lambda, \mu}$ and $\lambda:=\min_i\lambda_i,\  \mu:=\min_j \mu(\psi_j)$. 
\end{cor}


Now, let us apply our results to some restricted quantum computation models, such as 
Clifford circuits with solely magic state inputs (CM) and Pauli-based measurement (PBC), 
which gives an efficient simulation of $O(n)$ measurement with 
high probability.

\textit{\textbf{Example 3}}---Theorem \ref{thm:main2}  implies the following result: 
for Clifford circuit $\mathcal{C}$ with input states $ \ket{T}^{\ot n}$ and measurement
on $k$ qubits in computational basis with $k\leq (1-\log_{2}(3/2))n\approx 0.415n$, there exists a classical algorithm 
to approximate the output probability up to $l_1$ norm $\delta$ in time 
$n^{O((2+\sqrt{2})\log(\sqrt{\alpha}/\delta))}$  for at least $1-\frac{2}{\alpha}$ fraction of  Clifford circuits $\mathcal{C}$, 
where $\mu(\ket{T})=1-\frac{1}{\sqrt{2}}$ and $\chi(\ket{T})=3$. This may be contrasted with the hardness result ruling out efficient classical sampling from this class of circuits \cite{yoganathan2018quantum}.

\textit{\textbf{Example 4}}---A Pauli-Based Computation (PBC) is defined as a sequence of measurement of 
some Pauli operators $P_i\in P^n$, where the measurement outcome is
$(-1)^{\sigma_i} $ with $\sigma_i\in \set{0,1} $ and the Pauli operators $\set{P_i}$ are 
commuting with each other. Here, the initial state is 
$\ket{T}$ (or $\ket{H}=\cos\frac{\pi}{8}\ket{0}+\sin\frac{\pi}{8}\ket{1}$, which is equivalent 
to $\ket{T}$ up to Clifford unitary \cite{BravyiPRX16}.).
After $k$ steps, the probability 
of outcome $P(\sigma_1,\ldots,\sigma_k)=\bra{T^{\ot n}}{\Pi}\ket{T^{\ot n}}$,
where $\Pi=2^{-k}\prod^k_{i=1}(I+(-1)^{\sigma_i}P_i)$.
Note that PBC was considered in the fault-tolerant implementation of quantum computation based on stabilizer codes, 
where the stabilizer codes provide a simple realization of nondestructive Pauli measurements  \cite{Gottesman1998,Steane1997}. 
Besides, it has been proved that the quantum computation based on Clifford+$T$ circuits can be simulated by PBC \cite{BravyiPRX16}. Thus, this implies that the output probability  $P(\sigma_1,\ldots,\sigma_k)$ is $\#\mathsf P$-hard to simulate. 
It has been shown that any PBC on $n$ qubits can be classically simulated in $2^{cn}poly(n)$ time with $c\approx 0.94$ \cite{BravyiPRX16}.
Here, Theorem \ref{thm:main2} implies that if the measurement steps $k\leq (1-\log_2(3/2))n\approx 0.415n$, then
there exists a classical algorithm 
to approximate the output probability up to $l_1$ norm $\delta$ in time 
$n^{O((2+\sqrt{2})\log(1/\delta))}$ for a large fraction of PBC.

\section{Conclusion}
In this work, we investigated the problem of evaluating the output probabilities of Clifford  circuits with nonstabilizer input states. 
First, we provided an efficient classical algorithm to approximate the output probability of the Clifford circuits with 
mixed input states and showed that the running time scales with the increase in the purity of input states.
Second, we showed that a modification of this algorithm gives an efficient classical simulation for
pure nonstabilizer states, under some restriction on the number of measured qubits that is determined by the Pauli rank of 
the input states. The Pauli rank we introduced in this work can be regarded as a good candidate for a
magic monotone. 
We showed that these two results have several implications in other restricted quantum computation models such as 
Clifford circuits with magic input states, Pauli-based computation and IQP circuits. 

\begin{acknowledgments}
K. B. thanks Xun Gao for introducing the tensor network representation of quantum circuits to him, and for fruitful discussions related to this topic. K.B. acknowledges the Templeton Religion Trust for the partial support of this research under grants TRT0159
and   Zhejiang University for the support of an Academic Award for
Outstanding Doctoral Candidates. D.E.K. is funded by EPiQC, an NSF Expedition in Computing, under grant CCF-1729369.
\end{acknowledgments}

\bibliographystyle{apsrev4-1}
 \bibliography{Nosyinp-lit}

\appendix
\widetext

\section{Proof of Theorem 1}\label{sec:mix}

\subsection{Efficient evaluation of Fourier coefficients}
First, let us define the Fourier transformation on a single qubit state,  inspired by \cite{Gao2018}. Given a single qubit state $\rho\in D(\complex^2)$, we can write it in terms of its Bloch sphere representation 
\begin{eqnarray}\label{eq:bl1}
\rho=\frac{1}{2}\left( \rho_{00}I+\rho_{10}X+\rho_{01}Z+\rho_{11}XZ \right),
\end{eqnarray}
where $\rho_{00}=1$ and $|\rho_{10}|^2+|\rho_{01}|^2+|\rho_{11}|^2\leq1$.

Given $a, b\in\mathbb{F}_2$, it is easy to verify that 
\begin{eqnarray}\label{eq:inFT_2}
X^bZ^a\rho Z^aX^b
=\frac{1}{2}\sum_{s, t\in \mathbb{F}_2}
(-1)^{sa+tb}\rho_{st}X^sZ^t.
\end{eqnarray}
Thus, we can define the Fourier transformation on the state $\rho$ as follows
\begin{eqnarray}\label{eq:FT_2}
\mathbb{E}_{a\in \mathbb{F}_2, b\in \mathbb{F}_2}X^bZ^a\rho Z^aX^b(-1)^{sa+tb}
=\frac{1}{2}\rho_{st}X^sZ^t.
\end{eqnarray}
Note that for $t=s=0$, the above Fourier transformation is equal to the completely depolarizing channel.  And the equation \eqref{eq:inFT_2}
is the inverse Fourier  transformation of \eqref{eq:FT_2}.

Given the input states 
$\proj{0}^{\ot n}\ot^m_{i=1} \rho_i$ with Clifford unitary $U$, the output probability $q(\vec{y})$ is 
\begin{eqnarray}
q(\vec{y})=\bra{\vec{y}}U\proj{0}^{\ot n}\ot^m_{i=1} \rho_i U^\dag \ket{\vec{y}},
\end{eqnarray}
for any $\vec{y}\in\mathbb{F}^{n+m}_2$.
Let us denote the Pauli operators $Z^{\vec{a}}:=\ot^m_{i=1} Z^{a_i}, X^{\vec{b}}:=\ot^m_{i=1} X^{b_i}$  for any 
$\vec{a},\vec{b}\in \mathbb{F}^m_2$ to be operators acting on the latter $m$ qubits. Now, let us insert  $X^{\vec{b}}Z^{\vec{a}}$ into the $m$ mixed states  as follows
\begin{eqnarray}
q_{\vec{a}, \vec{b}}(\vec{y})
=\bra{\vec{y}}U\proj{0}^{\ot n}\ot(X^{\vec{b}} Z^{\vec{a}}\ot^m_{i=1} \rho_i Z^{\vec{a}}X^{\vec{b}})U^\dag \ket{\vec{y}}
=\bra{\vec{y}}U\proj{0}^{\ot n}\ot (\ot^m_{i=1} X^{b_i}Z^{a_i}\rho_i Z^{a_i}X^{b_i}) U^\dag \ket{\vec{y}}.
\end{eqnarray}
Hence, the output probability $q(\vec{y})=q_{\vec{0}, \vec{0}}(\vec{y})$.
Then, let us  take  the  Fourier transformation with respect to $\vec{a}, \vec{b}$ and the corresponding 
Fourier coefficient is
\begin{eqnarray*}
\hat{q}_{\vec{s},\vec{t}}&:=&\mathbb{E}_{\vec{s}\in\mathbb{F}^m_2, \vec{t}\in \mathbb{F}^m_2}
q_{\vec{a}, \vec{b}}(\vec{y})(-1)^{\vec{s}\cdot\vec{a}+\vec{t}\cdot\vec{b}}\\
&=&\mathbb{E}_{\vec{s}\in\mathbb{F}^m_2, \vec{t}\in \mathbb{F}^m_2}\bra{\vec{y}}U\proj{0}^{n}\ot (\ot^m_{i=1} X^{b_i}Z^{a_i}\rho_i Z^{a_i}X^{b_i}) U^\dag \ket{\vec{y}}\\
&=&\bra{\vec{y}}U\proj{0}^{\ot n}\ot (\ot^m_{i=1} \mathbb{E}_{a_i\in \mathbb{F}_2, b_i\in\mathbb{F}_2}X^{b_i}Z^{a_i}\rho_i Z^{a_i}X^{b_i}) U^\dag \ket{\vec{y}}.
\end{eqnarray*}
By equation \eqref{eq:FT_2}, we have 
\begin{eqnarray}
\hat{q}_{\vec{s},\vec{t}}
=\bra{\vec{y}}U\proj{0}^{\ot n}\ot^m
_{i=1}X^{s_i}Z^{t_i} U^\dag \ket{\vec{y}}\cdot \prod^m_{i=1}\left(\frac{\rho^{(i)}_{s_it_i}}{2}\right),
\end{eqnarray}
where $\rho^{(i)}_{s_it_i}$ is the coefficient of $\rho_i$ in the corresponding Bloch sphere representation.
Since $U$ is a Clifford unitary, then 
\begin{eqnarray*}
U\proj{0}^{\ot n}\otimes^m_{i=1}X^{s_i}Z^{t_i} U^\dag
=\prod^{n}_{i=1}\left(\frac{I+P_i}{2}\right)\prod^{m}_{j=1}Q_j,
\end{eqnarray*}
where the Pauli operators $P_i:=UZ_iU^\dag$ for $1\leq i\leq n$ and $P_j:=UX^{s_j}Z^{t_j} U^\dag $ for $1\leq j\leq m$  and they are commuting with each other.  Thus, by Gottesman-Knill Theorem,  the Fourier coefficients 
$\hat{q}_{\vec{s},\vec{t}}$  
can be evaluated in classical $O((n+m)^3)$ time .

\subsection{Exponential decay of Fourier coefficients}

Since $\rho $ is a mixed state in $D(\complex^2)$, it can always be written as $\rho=(1-\lambda)\sigma+\frac{\lambda}{2} I$, where
$\sigma$ is a pure state and $\lambda=1-\sqrt{2\Tr{\rho^2}-1}$. The pure state $\sigma$ also has the Bloch sphere representation 
\begin{eqnarray}\label{eq:bl2}
\sigma=\frac{1}{2}\left( \sigma_{00}I+\sigma_{10}X+\sigma_{01}Z+\sigma_{11}XZ \right),
\end{eqnarray}
where $\sigma_{00}=1$ and $|\sigma_{10}|^2+|\sigma_{01}|^2+|\sigma_{11}|^2=1$.
We have the following relationship between the coefficients $\rho_{st}$ and $\sigma_{st}$ for any 
$s,t\in\mathbb{F}_2$.

\begin{lem}\label{lem:rel1}
Given a mixed state $\rho=(1-\lambda)\sigma+\frac{\lambda}{2} I$, where 
$\rho, \sigma$ has Bloch sphere representation given by \eqref{eq:bl1} and \eqref{eq:bl2} respectively, 
then 
we have
\begin{eqnarray}
\rho_{st}=(1-\lambda)^{w(s,t)}\sigma_{st},
\end{eqnarray}
for any $s,t\in \mathbb{F}_2$, where $w(s, t)$ is defined as
\begin{equation}\label{eq:w1}
w(s, t)=
             \begin{cases}
          0, & s=0, t=0    \\
             1, &  \text{otherwise}
             \end{cases}.
\end{equation}
\end{lem}
\begin{proof}

This is because 

\begin{eqnarray*}
\rho_{st}=
\Tr{X^s\rho Z^t}
=(1-\lambda)\Tr{X^s\sigma Z^t}+\lambda/2
\Tr{X^sZ^t}
=(1-\lambda)\sigma_{st}+\lambda\delta_{s,0}\delta_{t,0}
=(1-\lambda)^{w(s,t)}\sigma_{st},
\end{eqnarray*}
where  $w(s,t)$ is defined as \eqref{eq:w1}.

\end{proof}

Each mixed input state $\rho_i$ can be written as $\rho_i=(1-\lambda_i)\sigma_i+\frac{\lambda_i}{2}I$ where
$\sigma_i$ is a pure state. 
Consider the quantum circuit with   input state
$\proj{0}^{\ot n}\ot^m_{i=1} \sigma_i$ and  Clifford unitary $U$, the output probability $p(\vec{y})$ is equal to 
\begin{eqnarray}
p(\vec{y})=\bra{\vec{y}}U\proj{0}^{\ot n}\ot^m_{i=1} \sigma_i U^\dag \ket{\vec{y}}.
\end{eqnarray}
Similar to $q(\vec{y})$, we insert $X^{\vec{b}}Z^{\vec{a}}$ into the circuit and define $p_{\vec{a},\vec{b}}$ as follows
\begin{eqnarray}
p_{\vec{a}, \vec{b}}(\vec{y})
=\bra{\vec{y}}U\proj{0}^{\ot n}\ot (\ot^m_{i=1} X^{b_i}Z^{a_i}\sigma_i Z^{a_i}X^{b_i}) U^\dag \ket{\vec{y}}.
\end{eqnarray}
Then  the corresponding Fourier coefficient can also be expressed as follows,
\begin{eqnarray}
\hat{p}_{\vec{s},\vec{t}}
=\bra{\vec{y}}U\proj{0}^{\ot n}\ot^m_{i=1}X^{s_i}Z^{t_i} U^\dag \ket{\vec{y}}\cdot \prod^m_{i=1}\left(\frac{\sigma^{(i)}_{s_it_i}}{2}\right),
\end{eqnarray}
where $\sigma^{(i)}_{s_it_i}$ is the coefficient of $\sigma_i$ in the corresponding Bloch sphere representation. 
By Lemma \ref{lem:rel1}, it is easy to see that 
\begin{eqnarray}
|\hat{q}_{\vec{s},\vec{t}}|\leq (1-\lambda)^{w(\vec{s}, \vec{t})}|\hat{p}_{\vec{s},\vec{t}}|,
\end{eqnarray}
where $\lambda=\min_i\lambda_i$ and 
$w(\vec{s}, \vec{t})$ is defined as 
\begin{eqnarray}\label{eq:w}
w(\vec{s}, \vec{t}):=\sum_{i} w(s_i, t_i).
\end{eqnarray}

\subsection{Good approximation with respect to $l_1$ norm}
The following lemma regarding Clifford unitaries on $n$ qubits is necessary the proof,
\begin{lem}[\cite{Dankert2009}]\label{lem:2-des}
The uniform distribution of Clifford unitaries on $n$ qubits is an exact 2-design, that is, for any 
$A, B, W$, we have 
\begin{eqnarray}
\mathbb{E}_{U\sim \mathcal{C}l_n}
U^{\dag} AU W U^{\dag} BU
=\int_{U(2^n)}\mathrm{d}U \ U^{\dag} AUWU^{\dag} BU,
\end{eqnarray}
where $\mathbb{E}_{U\sim \mathcal{C}l_n}:=\frac{1}{|\mathcal{C}l_n|}\sum_{U\sim \mathcal{C}l_n}$ and
\begin{eqnarray}
\int_{U(2^n)}\mathrm{d}U \ U^{\dag} AUWU^{\dag} BU
=\frac{\Tr{AB}\Tr{W}}{2^n}\frac{I}{2^n}
+\frac{2^n\Tr{A}\Tr{B}-\Tr{AB}}{2^n(2^{2n}-1)}\left(W-\Tr{W}\frac{I}{2^n}\right).
\end{eqnarray}
\end{lem}
Now, let us prove Theorem 1. Let us define
\begin{equation}
\hat{q}'_{\vec{s},\vec{t}}(\vec{y})=
             \begin{cases}
             \hat{q}_{\vec{s}, \vec{t}}(\vec{y}), & w(\vec{s},\vec{t})\leq l   \\
             0, &  \text{otherwise}
             \end{cases},
\end{equation}
which gives an family of unnormalized probability distribution $\set{q'_{\vec{a},\vec{b}}}$as
$q'_{\vec{a},\vec{b}}(\vec{y})=\sum_{\vec{s}, \vec{t} }\hat{q}'_{\vec{s},\vec{t}}(\vec{y})(-1)^{\vec{s}\cdot \vec{a}+\vec{t}\cdot\vec{b}}$ for each output $\vec{y}\in\mathbb{F}^{n+m}_2$
Then we show that $q'_{\vec{0}, \vec{0}}(\vec{y})$ gives a good approximation of $q_{\vec{0}, \vec{0}}(\vec{y})$ with respect to $l_1$ norm
\begin{eqnarray*}
\norm{q'_{\vec{0},\vec{0}}- q_{\vec{0},\vec{0}}}_1
=\sum_{\vec{y}\in \mathbb{F}^{n+m}_2}
|q'_{\vec{0}, \vec{0}}(\vec{y})-q_{\vec{0}, \vec{0}}(\vec{y})|
\end{eqnarray*}
 for a large fraction of 
Clifford circuits. 
First, since    $ \hat{q}_{\vec{s}, \vec{t}}(\vec{y})$  depends on the Clifford unitaries $U$, denote it as $\hat{q}_{\vec{s}, \vec{t}}(\vec{y})[U]$,  then it is easy to show that 
\begin{eqnarray}\label{eq:cli_eq}
 \hat{q}_{\vec{s}, \vec{t}}(\vec{y})[U](-1)^{\vec{a}\cdot \vec{s}+\vec{b}\cdot \vec{t}}
= \hat{q}_{\vec{s}, \vec{t}}(\vec{y})[U'],
\end{eqnarray}
where $U'=U\circ Z^{\vec{a}}X^{\vec{b}}$ is also a Clifford unitary for any $\vec{a},\vec{b}\in \mathbb{F}^m_2$ and $Z^{\vec{a}}X^{\vec{b}} $ act on
the $n+1,\ldots, n+m$th qubits. 
Thus 

\begin{eqnarray}\label{eq:ave_u}
\mathbb{E}_{U\sim \mathcal{C}l_{n+m}}\norm{q'_{\vec{0},\vec{0}}- q_{\vec{0},\vec{0}}}^2_1
=\mathbb{E}_{U\sim \mathcal{C}l_{n+m}}\norm{q'_{\vec{a}, \vec{b}}- q_{\vec{a}, \vec{b}}}^2_1
=\mathbb{E}_{U\sim \mathcal{C}l_{n+m}}\mathbb{E}_{\vec{a}\in\mathbb{F}^m_2, \vec{b}\in \mathbb{F}^m_2}\norm{q'_{\vec{a},\vec{b}}- q_{\vec{a},\vec{b}}}^2_1.
\end{eqnarray}
Moreover, 
\begin{eqnarray}
\nonumber\mathbb{E}_{\vec{a}\in\mathbb{F}^m_2, \vec{b}\in \mathbb{F}^m_2}\norm{q'_{\vec{a},\vec{b}}- q_{\vec{a},\vec{b}}}^2_1
\nonumber&\leq&  \mathbb{E}_{\vec{a}\in\mathbb{F}^m_2, \vec{b}\in \mathbb{F}^m_2} 2^{n+m}\sum_{\vec{y}\in\mathbb{F}^{n+m}_2}(q'_{\vec{a},\vec{b}}(\vec{y})- q_{\vec{a},\vec{b}}(\vec{y}))^2\\
\nonumber&=&2^{n+m}\sum_{\vec{y}\in\mathbb{F}^{n+m}_2}\mathbb{E}_{\vec{a}\in\mathbb{F}^m_2, \vec{b}\in \mathbb{F}^m_2} (q'_{\vec{a},\vec{b}}(\vec{y})- q_{\vec{a},\vec{b}}(\vec{y}))^2\\
\nonumber&=&2^{n+m}\sum_{\vec{y}\in\mathbb{F}^{n+m}_2}\sum_{\vec{s}\in\mathbb{F}^m_2,\vec{t}\in\mathbb{F}^m_2} (\hat{q}'_{\vec{s},\vec{t}}(\vec{y})- \hat{q}_{\vec{s},\vec{t}}(\vec{y}))^2\\
\nonumber&\leq&2^{n+m}(1-\lambda)^{2l}\sum_{\vec{y}\in\mathbb{F}^{n+m}_2}\sum_{w(\vec{s}, \vec{t})\geq l} \hat{p}^2_{\vec{s},\vec{t}}(\vec{y})\\
\nonumber&\leq &2^{n+m}(1-\lambda)^{2l}\sum_{\vec{y}\in\mathbb{F}^{n+m}_2}\sum_{\vec{s}\in\mathbb{F}^m_2,\vec{t}\in\mathbb{F}^m_2} \hat{p}^2_{\vec{s}, \vec{t}}(\vec{y})\\
\label{eq:For_a}&=&2^{n+m}(1-\lambda)^{2l}\sum_{\vec{y}\in\mathbb{F}^{n+m}_2}\mathbb{E}_{\vec{a}\in\mathbb{F}^m_2, \vec{b}\in \mathbb{F}^m_2} p^2_{\vec{a}, \vec{b}}(\vec{y}),
\end{eqnarray}
where the first line comes from the Cauchy-Schwarz inequality, the third line comes from the Parseval identity, and the fourth 
line comes from the fact that 
$|\hat{q}_{\vec{s},\vec{t}}(\vec{y})|\leq (1-\lambda)^{w(\vec{s}, \vec{t})}|\hat{p}_{\vec{s},\vec{t}}(\vec{y})|$.
According to Lemma \ref{lem:2-des}, we have 
\begin{eqnarray*}
\mathbb{E}_{U\sim \mathcal{C}l_{n+m}}p^2_{\vec{a}, \vec{b}}(\vec{y})
\leq 2\cdot2^{-2(n+m)}.
\end{eqnarray*}
Thus 
\begin{eqnarray*}
\mathbb{E}_{U\sim \mathcal{C}l_{n+m}}\norm{q'_{\vec{0},\vec{0}}- q_{\vec{0},\vec{0}}}^2_1
\leq 2e^{-\lambda l}.
\end{eqnarray*}
By Markov's inequality, we have
\begin{eqnarray*}
\mathrm{Pr}_{U\sim  \mathcal{C}l_{n+m}}\left[\norm{q'_{\vec{0},\vec{0}}- q_{\vec{0},\vec{0}}}_1\leq \sqrt{\alpha}e^{-\lambda l}\right]\geq 1-\frac{2}{\alpha}.
\end{eqnarray*}
Therefore, to obtain the $l_1$ norm up to $\delta$, we need take
$l=O(\log(\sqrt{\alpha}/\delta)/\lambda)$ and evaluate the Fourier coefficients $\hat{q}'_{\vec{s}, \vec{t}}(\vec{y})$ with $w(\vec{s}, \vec{t})\leq l$,
where total amount of such Fourier coefficients  is $\sum_{i\leq l}3^i C^i_m\leq 3^lm^l$. Thus, 
 there exists a classical algorithm 
to approximate each output probability $q(\vec{y})$ in time $O((n+m)^3)m^l=(n+m)^{O(1)}m^{O(\log(\sqrt{\alpha}/\delta)/\lambda)}$ with $l_1$ norm less than $\delta$ for at least $1-\frac{2}{\alpha}$  fraction of
Clifford circuits.
Thus, we finish the proof of Theorem 1.

\subsection{Slightly beyond Clifford circuits }\label{subsec:slig_cl}

Now, let us consider the quantum circuit $\mathcal{C}=\mathcal{C}_1\circ V$ with input state 
$\proj{0}^{\ot n}\ot^m_{i=1} \rho_i$  and  the gates in circuits $\mathcal{C}_1$ taken from the 
set of Clifford gates on $n$ qubits $\mathcal{C}l_{n+m}$ and $V$ is taken from the third level of Clifford hierarchy
 $\mathcal{C}l^{(3)}_{m}$ acting on $n+1,\ldots,(n+m)$th qubits. 
 The proof of Corollary 2 is almost the same as that of Theorem 1. 
 We only need to show the corresponding Fourier coefficients of  $q_{\vec{a}, \vec{b}}$ 
 also can be evaluated in $O((n+m)^3)$ time, where
 \begin{eqnarray}
q_{\vec{a}, \vec{b}}(\vec{y})
=\bra{\vec{y}}UV\proj{0}^{\ot n}\ot (X^{\vec{b}} Z^{\vec{a}} \ot^m_{i=1}\rho_i Z^{\vec{a}}X^{\vec{b}})V^\dag U^\dag \ket{\vec{y}}
=\bra{\vec{y}}UV\proj{0}^{\ot n}\ot (\ot^m_{i=1} X^{b_i}Z^{a_i}\rho_i Z^{a_i}X^{b_i}) U^\dag V^{\dag} \ket{\vec{y}}.
\end{eqnarray}
and $V\in \mathcal{C}l^{(3)}_m, U\in \mathcal{C}l_{n+m}$. Then the Fourier coefficient 
$\hat{q}_{\vec{s}, \vec{t}}(\vec{y})$ is equal to 
\begin{eqnarray}
\hat{q}_{\vec{s},\vec{t}}
=\bra{\vec{y}}U V\proj{0}^{\ot n}\ot^m_{i=1}X^{s_i}Z^{t_i}  V^\dag U^\dag \ket{\vec{y}}\cdot \prod^m_{i=1}\left(\frac{\rho^{(i)}_{s_it_i}}{2}\right).
\end{eqnarray}

Since $V\in \mathcal{C}l^{(3)}_m$, then $V\ot^m_{i=1}X^{s_i}Z^{t_i}  V^\dag \in \mathcal{C}l_m$. Thus, 
\begin{eqnarray*}
\hat{q}_{\vec{s},\vec{t}}
=\bra{\vec{y}}U \proj{0}^{\ot n}U' \ket{\vec{y}}\cdot \prod^m_{i=1}\left(\frac{\rho^{(i)}_{s_it_i}}{2}\right).
\end{eqnarray*}
where $U, U'=V\ot^m_{i=1}X^{s_i}Z^{t_i}  V^\dag U^\dag$ are both Clifford unitaries. Thus, the Fourier coefficient
$\hat{q}_{\vec{s},\vec{t}}$ can also be evaluated in $O((n+m)^3)$ time by Gottesman-Knill Theorem. 
Therefore, it is easy to prove Corollary 2 by following the proof of Theorem 1.


\section{Efficient classical simualtion of  IQP circuits with noisy input states}\label{sec:IQP}
In this section, we will prove the following proposition in Example 2:
\begin{prop}
Given an IQP circuit $H^{\ot n}DH^{\ot n}$ with the diagonal unitaries chosen from the 
gate set $\set{CZ, Z, S, T}$, if there is depolarizing nosie acting on each input state, i.e.,  input state
is $((1-\epsilon)\proj{0}+\frac{\epsilon}{2}I)^{\ot n}$,  then there exists an efficient  classical algorithm 
to approximate the output probabilities up to $l_1$ norm $\delta$ in time 
$n^{O(\log(\sqrt{\alpha}/\delta)/\epsilon)}$ for at least $1-\frac{2}{\alpha}$ fraction of IQP circuits. 
\end{prop}
\begin{proof}

The proof is similar to that of Theorem 1. 
If the state $\rho$ has some specific form as $\rho=\frac{1}{2}(\rho_{0}I+\rho_{1}Z)$, then we can simplify the  Fourier transformation  \eqref{eq:FT_2}
as 
\begin{eqnarray}\label{eq:FT_d}
\mathbb{E}_{a\in\mathbb{F}_2}X^{a}\rho X^a(-1)^{as}
=\frac{1}{2}\rho_sZ^s.
\end{eqnarray}

Given an IQP circuit $H^{\ot n} D H^{\ot}$ with noisy input states 
$\rho^{\ot n}$, $\rho=(1-\epsilon)\proj{0}+\epsilon\frac{I}{2}$, and gates in $D$ chosen from the gate set $\set{CZ, Z, S, T}$, then the output probability $q(\vec{y})$ is equal to 
\begin{eqnarray}
q(\vec{y})=\bra{\vec{y}}H^{\ot n} D H^{\ot n} \rho^{\ot n} H^{\ot n } D H^{\ot n}  \ket{\vec{y}}.
\end{eqnarray}
Similar to the proof of Theorem 1, we  insert $X^{\vec{a}}$  into the circuits for any $\vec{a}\in\mathbb{F}^n_2$ and define 
$q_{\vec{a}}(\vec{y})$ as follows
\begin{eqnarray}
q_{\vec{a}}(\vec{y})=\bra{\vec{y}}H^{\ot n} D H^{\ot n}X^{\vec{a}}\rho^{\ot n} X^{\vec{a}} H^{\ot n } D H^{\ot n}  \ket{\vec{y}}
=\bra{\vec{y}}H^{\ot n} D H^{\ot n}\ot_i X^{a_i}\rho X^{a_i} H^{\ot n } D H^{\ot n}  \ket{\vec{y}}
\end{eqnarray}
Then let us take the Fourier transformation with respect to $\vec{a}$ and the corresponding 
Fourier coefficient is 
\begin{eqnarray}
\nonumber\hat{q}_{\vec{s}}(\vec{y})
&:=&\mathbb{E}_{\vec{a}\in\mathbb{F}^n_2}
q_{\vec{a}}(\vec{y})(-1)^{\vec{s}\cdot \vec{a}}\\
\nonumber&=&\mathbb{E}_{\vec{a}\in\mathbb{F}^n_2}\bra{\vec{y}}H^{\ot n} D H^{\ot n}\ot_i X^{a_i}\rho X^{a_i} H^{\ot n } D H^{\ot n}  \ket{\vec{y}}(-1)^{\vec{s}\cdot \vec{a}}\\
\nonumber&=&\bra{\vec{y}}H^{\ot n} D H^{\ot n}\ot_i (\mathbb{E}_{a_i\in \mathbb{F}_2} X^{a_i}\rho X^{a_i} (-1)^{a_is_i})H^{\ot n } D H^{\ot n}  \ket{\vec{y}}\\
\nonumber&=&\bra{\vec{y}}H^{\ot n} D H^{\ot n}\ot_i Z^{s_i} H^{\ot n } D H^{\ot n}  \ket{\vec{y}}\prod^n_{i=1}\left(\frac{\rho_{s_i}}{2}\right)\\
\label{eq:IQP_m}&=&\bra{\vec{y}}H^{\ot n} D H^{\ot n}\ot_i Z^{s_i} H^{\ot n } D H^{\ot n}  \ket{\vec{y}}\prod^n_{i=1}\left(\frac{(1-\epsilon)^{s_i}}{2}\right),
\end{eqnarray}
where the second last equality comes from \eqref{eq:FT_d}.

Besides, 
\begin{eqnarray}
DH^{\ot n}\ot_i Z^{s_i} H^{\ot n }D^\dag
=D\ot_i X^{s_i} D^\dag
=D' \ot_iT^{\gamma_i} Z^{s_i} X^{-\gamma_i} D'^{\dag}
\end{eqnarray}
where the diagonal part $D$ can be written as $D'\circ \ot^n_{i=1}T^{\gamma_i} $ with $\gamma_i\in \mathbb{F}_2$ and 
the gates in $D'$ chosen from the gate set $\set{CZ, Z, S}$. 
It is easy to verify that
\begin{eqnarray}
T^{\gamma_i}X^{s_i}T^{-\gamma_i}
=e^{-i\frac{\pi}{4}\gamma_is_i}S^{\gamma_is_i} X^{s_i},
\end{eqnarray}
for any $\gamma_i, s_i\in\set{0,1}$. 
That is,  $DH^{\ot n}\ot_i Z^{s_i} H^{\ot n }D^\dag$ is a Clifford circuit.
Thus, each Fourier coefficient can be evaluated in $O(n^3)$ by Gottesman-Knill Theorem. 

We also consider the same IQP circuits with input states $ \proj{0}^{\ot n}$, 
then output probability $p(\vec{y})=\bra{\vec{y}}H^{\ot n} D H^{\ot n} \proj{0}^{\ot n} H^{\ot n } D H^{\ot n}  \ket{\vec{y}}$.
Similarly, we insert the operator $X^{\vec{a}}$ as follows
\begin{eqnarray}
p_{\vec{a}}(\vec{y})=\bra{\vec{y}}H^{\ot n} D H^{\ot n} X^{\vec{a}}\proj{0}^{\ot n} X^{\vec{a}} H^{\ot n } D H^{\ot n}  \ket{\vec{y}}.
\end{eqnarray}
And the corresponding Fourier coefficient 
is 
\begin{eqnarray}\label{eq:IQP_p}
\hat{p}_{\vec{s}}(\vec{y}):=
\mathbb{E}_{\vec{a}\in\mathbb{F}^n_2}p_{\vec{a}}(\vec{y})
=\bra{\vec{y}}H^{\ot n} D H^{\ot n}\ot_i Z^{s_i} H^{\ot n } D H^{\ot n}  \ket{\vec{y}}\cdot 2^{-n}.
\end{eqnarray}
Comparing \eqref{eq:IQP_m} with \eqref{eq:IQP_p}, we have the following relation
\begin{eqnarray}
\hat{q}_{\vec{s}}(\vec{y})
=(1-\epsilon)^{|\vec{s}|}\hat{p}_{\vec{s}}(\vec{y}),
\end{eqnarray}
where $|\vec{s}|=\sum_is_i$ is the Hamming weight of $\vec{s}\in \mathbb{F}^n_2$.

Let us define
\begin{equation}
\hat{q}'_{\vec{s}}(\vec{y})=
             \begin{cases}
             \hat{q}'_{\vec{s}}(\vec{y}), & |\vec{s}|\leq l   \\
             0, &  \text{otherwise}
             \end{cases},
\end{equation}
which gives an family of unnormalized probability distribution $\set{q'_{\vec{a}}}$as
$q'_{\vec{a}}(\vec{y})=\sum_{\vec{s}}\hat{q}'_{\vec{s}}(\vec{y})(-1)^{\vec{s}\cdot \vec{a}}$ for each output $\vec{y}\in\mathbb{F}^n_2$.
Then we will show that $q'_{\vec{0}}(\vec{y})$ gives a good approximation of $q_{\vec{0}}(\vec{y})$ with respect to $l_1$ norm
\begin{eqnarray*}
\norm{q'_{\vec{0}}- q_{\vec{0}}}_1
=\sum_{\vec{y}\in \mathbb{F}^n_2}
|q'_{\vec{0}}(\vec{y})-q_{\vec{0}}(\vec{y})|
\end{eqnarray*}
 for a large fraction of IQP circuits. We denote $\mathbb{D}_n$ to be the set of of 
 diagonal part of IQP circuits where
 the diagonal gates are chosen from 
 $\set{ CZ, Z, S, T}$. Since    $ \hat{q}_{\vec{s}}(\vec{y})$  depends on the IQP circuits, denote it as $\hat{q}_{\vec{s}, \vec{t}}(\vec{y})[D]$,  then it is easy to verify that 
\begin{eqnarray*}
 \hat{q}_{\vec{s}}(\vec{y})[D](-1)^{\vec{a}\cdot \vec{s}}
= \hat{q}_{\vec{s}}(\vec{y})[D'],
\end{eqnarray*}
where $D'=D\circ Z^{\vec{a}}$ also belongs to $\mathbb{D}_n$.
Thus 

\begin{eqnarray*}
\mathbb{E}_{D\sim \mathbb{D}_n}\norm{q'_{\vec{0}}- q_{\vec{0}}}^2_1
=\mathbb{E}_{D\sim \mathbb{D}_n}\norm{q'_{\vec{a}}- q_{\vec{a}}}^2_1
=\mathbb{E}_{D\sim \mathbb{D}_n}\mathbb{E}_{\vec{a}\in\mathbb{F}^n_2}\norm{q'_{\vec{a}}- q_{\vec{a}}}^2_1.
\end{eqnarray*}
And
\begin{eqnarray*}
\mathbb{E}_{\vec{a}\in\mathbb{F}^n_2}\norm{q'_{\vec{a}}- q_{\vec{a}}}^2_1
&\leq&  \mathbb{E}_{\vec{a}\in\mathbb{F}^n_2} 2^n\sum_{\vec{y}\in\mathbb{F}^n_2}(q'_{\vec{a}}(\vec{y})- q_{\vec{a}}(\vec{y}))^2\\
&=&2^n\sum_{\vec{y}\in\mathbb{F}^n_2}\mathbb{E}_{\vec{a}\in\mathbb{F}^n_2} (q'_{\vec{a}}(\vec{y})- q_{\vec{a}}(\vec{y}))^2\\
&=&2^n\sum_{\vec{y}\in\mathbb{F}^n_2}\sum_{\vec{s}\in\mathbb{F}^n_2} (\hat{q}'_{\vec{s}}(\vec{y})- \hat{q}_{\vec{s}}(\vec{y}))^2\\
&\leq&2^n(1-\epsilon)^{2l}\sum_{\vec{y}\in\mathbb{F}^n_2}\sum_{|\vec{s}|\geq l} \hat{p}^2_{\vec{s}}(\vec{y})\\
&\leq &2^n(1-\epsilon)^{2l}\sum_{\vec{y}\in\mathbb{F}^n_2}\sum_{\vec{s}\in\mathbb{F}^n_2} \hat{p}^2_{\vec{s}}(\vec{y})\\
&=&2^n(1-\epsilon)^{2l}\sum_{\vec{y}\in\mathbb{F}^n_2}\mathbb{E}_{\vec{a}\in\mathbb{F}^n_2} p^2_{\vec{a}}(\vec{y}),
\end{eqnarray*}
where the first line comes from the Cauchy-Schwarz inequality, the third line comes from Parvesal identity, 
and the fourth line comes from the fact that $\hat{q}_{\vec{s}}(\vec{y})=(1-\epsilon)^{|\vec{s}|}\hat{p}_{\vec{s}}(\vec{y})$.
According to Lemma \ref{lem:av_dis} in Appendix \ref{sec:dis_IQP}, we have 
\begin{eqnarray*}
\mathbb{E}_{D\sim \mathbb{D}_n}\sum_{\vec{y}\in\mathbb{F}^n_2}p^2_{\vec{a}}(\vec{y})
\leq 2^{-(n-1)}.
\end{eqnarray*}
Thus,
we have 
\begin{eqnarray*}
\mathbb{E}_{D\sim \mathbb{D}_n}\norm{q'_{\vec{0}}- q_{\vec{0}}}^2_1
\leq 2e^{-2\epsilon l}.
\end{eqnarray*}
Therefore, by Markov's inequality, we have 
\begin{eqnarray*}
\mathrm{Pr}_{D\sim  \mathbb{D}_n}\left[\norm{q'_{\vec{0}}- q_{\vec{0}}}_1\leq \sqrt{\alpha}e^{-\epsilon l}\right]\geq 1-\frac{2}{\alpha}.
\end{eqnarray*}

Therefore, to obtain the $l_1$ norm up to $\delta$, we need 
take $l=O(\log(\sqrt{\alpha}/\delta)/\epsilon)$ and the total computational complexity 
is 
$O(n^3n^l)=n^{O(\log(\sqrt{\alpha}/\delta)/\epsilon)}$.
\end{proof}

\section{Distribution of IQP circuits based on Gowers uniformity norm}\label{sec:dis_IQP}
Here we consider IQP circuits, which can be represented by $H^{\ot n} D H^{\ot n}\ket{0}^{\ot n}$, where the gates in the diagonal part $D$ are chosen from the gate set
$\set{CZ, Z, S, T}$. Then the output distribution is
$p(\vec{y})=|\bra{\vec{y}}H^{\ot n} D H^{\ot n}\ket{0}^{\ot n}|^2=|\hat{f}(\vec{y})|^2$ for any $\vec{y}\in \mathbb{F}^n_2$, where
$\hat{f}(\vec{y})=\frac{1}{2^n}\sum_{\vec{x}\in \mathbb{F}^n_2}f(\vec{x})(-1)^{\vec{y}\cdot \vec{x}}$
and the function $f$ can be expressed as
\begin{eqnarray}
f(\vec{x})
=(-1)^{\sum_{i<j} \alpha_{ij}x_ix_j+\sum_i\beta_ix_i}
i^{\sum_i \gamma_i x_i}
e^{i\pi/4 \sum_it_ix_i},
\end{eqnarray}
where $\alpha_{ij}, \beta_i, \gamma_i, t_i\in \mathbb{F}_2$,  denote the number 
of $CZ$ between $i$th and $j$th qubits, Z gate on $i$th qubit, S gate on $i$th gate and T gate on $i$th gate. Since 
$T^2=S, S^2=Z$ and $Z^2=I$, then there are at most one $T$, $S$, $Z$ gate on each qubit respectively. 
Thus, $\vec{\beta}, \vec{\gamma}, \vec{t}\in \mathbb{F}^n_2$ and the Hamming weight 
$|\vec{\beta}|, |\vec{\gamma}|, |\vec{t}|$ is the number of $Z$, $S$ and $T$ gates in the IQP circiut.

In fact, the function $f$ can be rewritten as follows
\begin{eqnarray}
f(\vec{x})
=(-1)^{\vec{\beta}\cdot \vec{x}}i^{\vec{x}A\vec{x}}
e^{i\pi/4\vec{t}\cdot \vec{x}},
\end{eqnarray}
where $A_{ii}=\gamma_i$ and $A_{ij}=A_{ji}=\alpha_{ij}$ for $i\neq j$. That is, the matrix $A$ is a
symmetric $0-1$ matrix.

Now, let us introduce the Gowers uniformity norm here. 
Let $G$ be a finite additive group and  $f:G\to \complex$  and an integer $d\geq1$. Then the 
Gowers uniformity norm $\norm{f}_{U^d(G)}$ \cite{tao2006additive} is defined as
\begin{eqnarray}
\norm{f}^{2^d}_{U^d(G)}
=\mathbb{E}_{h_1,..,h_d,x\in G}
\Delta_{h_1}...\Delta_{h_d}f(x),
\end{eqnarray}
where $\Delta_hf(x):=f(x+h)\overline{f(x)}$. 
Here we take $G=\mathbb{F}^n_2$ and the Fourier transformation for $f:\mathbb{F}^n_2\to \complex$ 
is defined as $\hat{f}(\vec{y})=\mathbb{E}_{\vec{x}\in\mathbb{F}^n_2}f(x)(-1)^{\vec{x}\cdot \vec{y}}$, 
where $\mathbb{E}_{\vec{x}\in\mathbb{F}^n_2}:=\frac{1}{2^n}\sum_{\vec{x}\in\mathbb{F}^n_2}$. 
One important property of Gowers uniformity norm, which we will use in the following 
section to demonstrate the distribution of IQP circuits, is the following equality \cite{tao2006additive}
\begin{eqnarray}\label{lem:eq_2}
\norm{f}^4_{U^2(\mathbb{F}^n_2)}
=\sum_{\vec{y}\in \mathbb{F}^n_2}
|\hat{f}(\vec{y})|^4.
\end{eqnarray}

For IQP circuits with diagonal gates chosen from 
$\set{CZ, Z, CCZ}$ randomly, it has been proved that the average value of the second moment of  output probability
satisfies that
$\sum_{\vec{y}}p^2_D(\vec{y})\leq \alpha 2^{-n}$, where $\alpha$ is some constant \cite{BremnerPRL2016}. 
Here, we consider the case where
the gates in the diagonal part $D$ are chosen uniformly, 
i.e., 
$P(\alpha_{ij=1})=P(\beta_i=1)=P(\gamma_i=1)=P(t_i=1)=1/2$, then 
we can give the exact value of average value of the second moment of the output probability of random IQP circuits. 
\begin{lem}\label{lem:av_dis}
Given an IQP circuit, if the gates in the diagonal part $D$ can be chosen uniformly, then 
\begin{eqnarray}
\mathbb{E}_{D}\sum_{\vec{y}\in \mathbb{F}^n_2}p^2_D(\vec{y})
=2^{-(n-1)}-2^{-2n}.
\end{eqnarray}
\end{lem}
\begin{proof}
Due to the equation \eqref{lem:eq_2}, we have
\begin{eqnarray}
\sum_{\vec{y}\in\mathbb{F}^n_2}
|p(\vec{y})|^2
=\sum_{\vec{s}\in \mathbb{F}^n_2}
|\hat{f}(\vec{s})|^4
=\norm{f}^4_{U^2(\mathbb{F}^n_2)}.
\end{eqnarray}
 For the function $f(\vec{x})
=(-1)^{\vec{\beta}\cdot \vec{x}}i^{\vec{x}A\vec{x}}
e^{i\pi/4\vec{t}\cdot \vec{x}}$, the Gowers uniformity norm  
$\norm{f}_{U^2(\mathbb{F}^n_2)}$ can be expressed as follows

\begin{eqnarray*}
\norm{f}^4_{U^2(\mathbb{F}^n_2)}
&=&\mathbb{E}_{\vec{a}, \vec{b}, \vec{x}\in \mathbb{F}^n_2}
f(\vec{x}\oplus\vec{a}\oplus\vec{b})
f(\vec{x})\overline{f(\vec{x}\oplus \vec{a})f(\vec{x}\oplus \vec{b})}\\
&=&\mathbb{E}_{\vec{a}, \vec{b}, \vec{x}\in \mathbb{F}^n_2} i^{2\vec{a}A\vec{b}}
e^{i\pi/4 \sum_i t_i[(x_i\oplus a_i\oplus b_i)+x_i-(x_i\oplus a_i)-(x_i\oplus b_i)]}\\
&=&\mathbb{E}_{\vec{a}, \vec{b}, \vec{x}\in \mathbb{F}^n_2}(-1)^{\vec{a}A\vec{b}}
e^{i\pi/4 \sum_i t_i[(x_i\oplus a_i\oplus b_i)+x_i-(x_i\oplus a_i)-(x_i\oplus b_i)]}\\\
&=&\mathbb{E}_{\vec{a}, \vec{b} \in \mathbb{F}^n_2}
(-1)^{\vec{a}A\vec{b}}
\mathbb{E}_{ \vec{x}\in \mathbb{F}^n_2}
e^{i\pi/4 \sum_i t_i[(x_i\oplus a_i\oplus b_i)+x_i-(x_i\oplus a_i)-(x_i\oplus b_i)]}\\
&=&\mathbb{E}_{\vec{a}, \vec{b}\in \mathbb{F}^n_2}(-1)^{\vec{a}A\vec{b}}
\prod^n_{i=1}
\mathbb{E}_{x_i\in \mathbb{F}_2}
e^{i\pi/4 t_i[(x_i\oplus a_i\oplus b_i)+x_i-(x_i\oplus a_i)-(x_i\oplus b_i)]}.
\end{eqnarray*}

 It is easy to verify that 
\begin{eqnarray*}
\mathbb{E}_{x\in \mathbb{F}_2}e^{i\pi/4 t[(x\oplus a\oplus b)+x-(x\oplus a)-(x\oplus b)]}
=\frac{1+(-1)^{tab}}{2},
\end{eqnarray*}
for any $t, a, b\in \mathbb{F}_2$. Thus, we have 
\begin{eqnarray*}
\norm{f}^4_{U^2(\mathbb{F}^n_2)}=\mathbb{E}_{\vec{a}, \vec{b}\in \mathbb{F}^n_2}
(-1)^{\vec{a}A\vec{b}}
\prod^n_{i=1}\left[ \frac{1+(-1)^{t_ia_ib_i}}{2}  \right].
\end{eqnarray*}

The expected value of $\sum_{\vec{y}\in \mathbb{F}^n_2}p^2_D(\vec{y})$ over the random IQP circuits is
\begin{eqnarray*}
&&\mathbb{E}_{D}\sum_{\vec{y}\in\mathbb{F}^n_2}p^2_D(\vec{y})\\
&=&\mathbb{E}_{D}\norm{f_D}^4_{U^2(\mathbb{F}^n_2)}\\
&=&\mathbb{E}_{\set{\alpha_{ij}, \beta_i,\gamma_i, t_i}}\mathbb{E}_{\vec{a}, \vec{b}\in \mathbb{F}^n_2}
(-1)^{\vec{a}A\vec{b}}
\prod^n_{i=1}\left[ \frac{1+(-1)^{t_ia_ib_i}}{2}  \right]\\
&=&\mathbb{E}_{\vec{a}, \vec{b}\in \mathbb{F}^n_2}
\mathbb{E}_{\set{\alpha_{ij}, \beta_i,\gamma_i, t_i}}
(-1)^{\vec{a}A\vec{b}}
\prod^n_{i=1}\left[ \frac{1+(-1)^{t_ia_ib_i}}{2}  \right]\\
&=&\mathbb{E}_{\vec{a}, \vec{b}\in \mathbb{F}^n_2}
\prod_{i<j}\left[ \frac{1+(-1)^{a_ib_j+b_ia_j}}{2}\right]
\prod^n_{i=1}\left[ \frac{1+(-1)^{a_ib_i}}{2}\right]
\left[ \frac{3+(-1)^{a_ib_i}}{4}\right].\\
\end{eqnarray*}
Since
\begin{equation*}
\frac{1+(-1)^{a_i, b_i}}{2}=
             \begin{cases}
             0, &  (a_i, b_i)=(1,1)   \\
              1, &  \text{otherwise}
             \end{cases},
\end{equation*}
then the above  equation is equal to 
\begin{eqnarray*}
\frac{1}{4^n}\sum_{\times^n_{i=1}(a_i, b_i)\in\set{(0,0), (0,1), (1,0)}^{\times n}}
\prod_{i<j}\left[ \frac{1+(-1)^{a_ib_j+b_ia_j}}{2}\right]
=\frac{1}{4^n}\sum_{\times^n_{i=1}(a_i, b_i)\in\set{(0,0), (0,1), (1,0)}^{\times n}}
\prod_{i<j}\left[ \frac{1+(-1)^{(a_i+a_j)(b_i+b_j)}}{2}\right],
\end{eqnarray*}
where the equality comes from the fact that 
\begin{eqnarray}
a_ib_j+b_ia_j=(a_i+a_j)(b_i+b_j)-(a_ib_i+a_jb_j)
=(a_i+a_j)(b_i+b_j),
\end{eqnarray}
when $(a_i, b_i), (a_j, b_j)$ are chosen from 
$\set{(0,0), (0,1), (1,0)}$.
Moreover, for $(a_i, b_i), (a_j, b_j)\in\set{(0,0), (0,1), (1,0)}$, we have

\begin{equation*}
\frac{1+(-1)^{(a_i+a_j)(b_i+b_j)}}{2}=
             \begin{cases}
             0, &  (a_i, b_i, a_j, b_j)=(1,0,0,1), (0,1,1,0)   \\
              1, &  \text{otherwise}
             \end{cases}.
\end{equation*}
Thus, 
\begin{eqnarray*}
\sum_{\times^n_{i=1}(a_i, b_i)\in\set{(0,0), (0,1), (1,0)}^{\times n}}
\prod_{i<j}\left[ \frac{1+(-1)^{(a_i+a_j)(b_i+b_j)}}{2}\right]
=\left(\sum_{\times^n_{i=1}(a_i, b_i)\in\set{(0,0), (0,1)}^{\times n}}1\right)
+\left(\sum_{\times^n_{i=1}(a_i, b_i)\in\set{(0,0), (10)}^{\times n}}1\right)
-1
=2^{n+1}-1.
\end{eqnarray*}
Therefore, we obtain the result that 
\begin{eqnarray*}
\mathbb{E}_{D}\sum_sp^2_D(s)
=\frac{1}{4^n}[2^{n+1}-1].
\end{eqnarray*}

\end{proof}

Besides, based on the Gowers uniformity norm, we can also give an approximation of 
the second moment for any IQP circuit. 
\begin{prop}
Given an  IQP circuit with the diagonal gates chosen from $\set{CZ, Z, S, T}$, then the  output probability of this circuit satisfies the 
following property,
\begin{eqnarray}
\sum_{\vec{y}\in \mathbb{F}^n_2}p^2(\vec{y})
\leq 2^{-c|\vec{t}|-\mathrm{Rank} (A(\vec{t}))},
\end{eqnarray}
where the constant  $c=\log \frac{4}{3}>0$, 
$A(\vec{t})$ is the matrix obtained from $A$ by removing the rows  and columns $i$ such that $t_i=1$ and 
$\mathrm{Rank}  (A(\vec{t}))$ denotes the rank of the matrix $A(\vec{t}) $ in $\mathbb{F}_2$.
Moreover, if $\vec{t}=0$, i.e., there is no $T$ gate,  then
\begin{eqnarray}
\sum_{\vec{y}\in \mathbb{F}^n_2}p^2(\vec{y})
=2^{-\mathrm{Rank} (A) }.
\end{eqnarray}
\end{prop}

\begin{proof}
Due to the equation \eqref{lem:eq_2} and  Lemma \ref{lem:av_dis}, we have
\begin{eqnarray*}
\sum_{\vec{y}\in\mathbb{F}^n_2}
p^2(\vec{y})
=\sum_{\vec{s}\in \mathbb{F}^n_2}
|\hat{f}(\vec{s})|^4
=\norm{f}^4_{U^2(\mathbb{F}^n_2)}
=\mathbb{E}_{\vec{a}, \vec{b}\in \mathbb{F}^n_2}
(-1)^{\vec{a}A\vec{b}}
\prod^n_{i=1}\left[ \frac{1+(-1)^{t_ia_ib_i}}{2}  \right].
\end{eqnarray*}
Thus, we need estimate  the Gower uniform norm $\norm{f}_{U^2(\mathbb{F}^n_2)}$
for the phase polynomial $f(\vec{x})
=(-1)^{\vec{\beta}\cdot \vec{x}}i^{\vec{x}A\vec{x}}
e^{i\pi/4\vec{t}\cdot \vec{x}}$ by the Hamming weight $|\vec{t}|$ and 
the rank of the symmetric matrix $A$.

Without loss of generality, we assume the first $k=|\vec{t}|$ qubits have $T$ gates, i.e.,
$t_1=\ldots=t_k=1$, and the remaining qubits do not have $T$ gate, then we can decompose the symmetric matrix A as follows 
\begin{equation*}
A=\left[
\begin{array}{cc}
A_{k,k}& A_{k,n-k}\\
A_{n-k,k}  &A_{n-k, n-k}\\
\end{array}
\right],
\end{equation*}
where $A_{k,k}$ is an $k \times k$ symmetric matrix, 
$A_{n-k,n-k}$ is an $(n-k)\times (n-k)$ symmetric matrix and
$A_{n-k,k}=A^t_{k, n-k}$. Similarly, we also decompose 
the vectors $\vec{a}, \vec{b}$ as 
\begin{equation*}
\vec{a}=
\left[
\begin{array}{cc}
\vec{a}_k\\
\vec{a}_{n-k}
\end{array}
\right],
\vec{b}=
\left[
\begin{array}{cc}
\vec{b}_k\\
\vec{b}_{n-k}
\end{array}
\right],
\end{equation*}
where 
$\vec{a}_k, \vec{b}_k\in \mathbb{F}^k_2$ and 
$\vec{a}_{n-k}, \vec{b}_{n-k}\in \mathbb{F}^{n-k}_2$.
Thus, 
\begin{eqnarray*}
\norm{f}^4_{U^2(\mathbb{F}^n_2)}
&=&\mathbb{E}_{\vec{a}, \vec{b}\in \mathbb{F}^n_2}
(-1)^{\vec{a}A\vec{b}}
\left[\prod^k_{i=1} \frac{1+(-1)^{a_ib_i}}{2}   \right]\\
&=&
\mathbb{E}_{\vec{a}_k, \vec{b}_k\in \mathbb{F}^k_2}
(-1)^{\vec{a}_kA_{k,k}\vec{b}_k }
\left[\prod^k_{i=1} \frac{1+(-1)^{a_ib_i}}{2}   \right]\mathbb{E}_{\vec{a}_{n-k}, \vec{b}_{n-k}\in \mathbb{F}^{n-k}_2}
(-1)^{\vec{a}_{n-k}A_{n-k,n-k}\vec{b}_{n-k}+\vec{a}_kA_{k, n-k} \vec{b}_{n-k}+\vec{a}_{n-k}A_{n-k, k}\vec{b}_k}.
\end{eqnarray*}
Since
\begin{equation}
\frac{1+(-1)^{a_i, b_i}}{2}=
             \begin{cases}
             0, &  (a_i, b_i)=(1,1)   \\
              1, &  \text{otherwise}
             \end{cases},
\end{equation}
then
\begin{eqnarray*}
&&\left|\mathbb{E}_{\vec{a}_k, \vec{b}_k\in \mathbb{F}^k_2}
(-1)^{\vec{a}_kA_{k,k}\vec{b}_k }
\left[\prod^k_{i=1} \frac{1+(-1)^{a_ib_i}}{2}   \right]\mathbb{E}_{\vec{a}_{n-k}, \vec{b}_{n-k}\in\mathbb{F}^{n-k}_2}(-1)^{\vec{a}_{n-k}A_{n-k,n-k}\vec{b}_{n-k}+\vec{a}_kA_{k, n-k} \vec{b}_{n-k}+\vec{a}_{n-k}A_{n-k, k}\vec{b}_k}\right|\\
&=&
\left|\frac{1}{4^k}
\sum_{\times^k_{i=1}(a_i, b_i)\in \set{(0,0), (0,1), (1,0)}^{k}}
(-1)^{\vec{a}_kA_{k,k}\vec{b}_k}\mathbb{E}_{\vec{a}_{n-k}, \vec{b}_{n-k}\in\mathbb{F}^{n-k}_2}(-1)^{\vec{a}_{n-k}A_{n-k,n-k}\vec{b}_{n-k}+\vec{a}_kA_{k, n-k} \vec{b}_{n-k}+\vec{a}_{n-k}A_{n-k, k}\vec{b}_k}\right|\\
&\leq &
\frac{1}{4^k}\sum_{\times^k_{i=1}(a_i, b_i)\in \set{(0,0), (0,1), (1,0)}^{k}}
\left|
\mathbb{E}_{\vec{a}_{n-k}, \vec{b}_{n-k}}
(-1)^{\vec{a}_{n-k}A_{n-k,n-k}\vec{b}_{n-k}+\vec{a}_kA_{k, n-k} \vec{b}_{n-k}+\vec{a}_{n-k}A_{n-k, k}\vec{b}_k}
\right|\\
&\leq &
\left(\frac{3}{4}\right)^k \max_{\vec{x}, \vec{y}\in\mathbb{F}^{n-k}_2}
\left|
\mathbb{E}_{\vec{a}_{n-k},\vec{b}_{n-k}\in\mathbb{F}^{n-k}_2}
(-1)^{\vec{a}_{n-k}A_{n-k,n-k}\vec{b}_{n-k}+\vec{x}\cdot \vec{b}_{n-k}+\vec{y}\cdot\vec{a}_{n-k}}
\right|.
\end{eqnarray*}
Besides, for any $\vec{x}, \vec{y}\in\mathbb{F}^{n-k}_2$, 

\begin{eqnarray*}
&&
\left|\mathbb{E}_{\vec{a}_{n-k}, \vec{b}_{n-k}\in \mathbb{F}^{n-k}_2}
(-1)^{\vec{a}_{n-k}A_{n-k,n-k}\vec{b}_{n-k}+\vec{x}\cdot \vec{b}_{n-k}+\vec{y}\cdot\vec{a}_{n-k}}\right|\\
&=&
\left|\mathbb{E}_{\vec{a}_{n-k}\in \mathbb{F}^{n-k}_2}(-1)^{\vec{y}\cdot \vec{a}_{n-k}}
\mathbb{E}_{\vec{b}_{n-k}\in \mathbb{F}^{n-k}_2}(-1)^{(A_{n-k, n-k}\vec{a}_{n-k}+\vec{x})^T\vec{b}_{n-k}}
\right|\\
&= &
\left|\mathbb{E}_{\vec{a}_{n-k}\in \mathbb{F}^{n-k}_2}
\delta_{A_{n-k, n-k}\vec{a}_{n-k}, \vec{x}}
(-1)^{\vec{y}\cdot \vec{a}_{n-k}}
\right|\\
&\leq &
\frac{|\set{\vec{a}_{n-k}: A_{n-k,n-k}\vec{a}_{n-k}=\vec{x}}|}{2^{n-k}}\\
&\leq & \frac{1}{2^{n-k}} \left|\text{Ker} (A_{n-k, n-k})\right|\\
&=&\frac{1}{2^{\text{Rank} (A_{n-k, n-k})}},
\end{eqnarray*}
where $\text{Rank} (A_{n-k, n-k})$ denotes the rank of the matrix $A_{n-k,n-k}$
in $\mathbb{F}_2$. Therefore, 
\begin{eqnarray*}
\norm{f}^4_{U^2(\mathbb{F}^n_2)}
\leq\left( \frac{3}{4} \right)^k 
\frac{1}{2^{\text{Rank} (A_{n-k, n-k})}}
=2^{-ck-\text{Rank} (A_{n-k, n-k})},
\end{eqnarray*}
where $c=\log \frac{4}{3}$. 

Moreover, if $\vec{t}=0$, then
\begin{eqnarray*}
\norm{f}^4_{U^2(\mathbb{F}^n_2)}
=\mathbb{E}_{\vec{a}\in \mathbb{F}^n_2}\mathbb{E}_{ \vec{b}\in \mathbb{F}^n_2}
(-1)^{\vec{a}A\vec{b}}
=\mathbb{E}_{\vec{a}\in \mathbb{F}^n_2}
\delta_{A\vec{a}, \vec{0}}
=\frac{\text{Ker} (A)}{2^n}
=2^{-\text{Rank} (A)}.
\end{eqnarray*}

\end{proof}

\section{Efficent classical simulation with pure nonstabilizer input states}\label{sec:non_stab}

\subsection{Proof of Theorem 3}

\begin{lem}
For any pure state $\ket{\psi}$ in $D(\complex^2)$, the stabilizer fidelity can be expressed as
\begin{eqnarray}
F(\psi)=\frac{1}{2}\left(1+\max_{P\in \set{X, Y, Z}} |\bra{\psi}P\ket{\psi}|\right).
\end{eqnarray}
\end{lem}
\begin{proof}
This follows directly from the fact the single-qubit stabilizer states are the eigenstates of $X, Y, Z$, that is,
the stabilizer states have the form
$\proj{\phi}=\frac{I\pm P}{2}$, where $P\in\set{X, Y, Z}$.
\end{proof}

Thus $\mu(\psi)$ can also be expressed as
\begin{eqnarray}
\mu(\psi)=1-\max_{P\in \set{X, Y, Z}} |\bra{\psi}P\ket{\psi}|.
\end{eqnarray}

Now, let us begin the proof of Theorem 3.
Since $\ket{\psi}$ has  the Bloch sphere representation as $\proj{\psi}=\frac{1}{2}\sum_{s,t\in\mathbb{F}_2}\psi_{st}X^{s}Z^t$, it is easy to see that
\begin{eqnarray}\label{eq:p_p}
|\psi_{st}|\leq 1-\mu(\psi). 
\end{eqnarray}
for any $(s, t)\neq (0,0)$.
 
Without loss of generality, we assume the first $k$ qubits are measured as the swap gate belongs to $\mathcal{C}l_{n+m}$. 
Then the output probability  is 
\begin{eqnarray}
q(\vec{y})=\Tr{U\proj{0}^{\ot n} \ot^m_{i=1}\proj{\psi_i} U^\dag \proj{\vec{y}}\ot I_{n+m-k}},
\end{eqnarray}
for any $\vec{y}\in \mathbb{F}^k_2$, where $I_{n+m-k}$ denotes the identity on the $k+1,..., (n+m)$th qubits. 
Let us  insert the Pauli operator $X^{\vec{b}}Z^{\vec{a}}$ into the circuit and the corresponding output probability 
\begin{eqnarray}
q_{\vec{a}, \vec{b}}(\vec{y})=\Tr{U\proj{0}^{\ot n}\ot (X^{\vec{b}}Z^{\vec{a}}\ot^m_{i=1}\proj{\psi_i}Z^{\vec{a}} X^{\vec{b}})U^\dag \proj{\vec{y}}\ot I_{n+m-k}}.
\end{eqnarray}
The corresponding Fourier coefficient is
\begin{eqnarray}
\hat{q}_{\vec{s}, \vec{t}}(\vec{y})
=\mathbb{E}_{\vec{a}\in\mathbb{F}^m_2, \vec{b}\in\mathbb{F}^m_2}q_{\vec{a}, \vec{b}}(\vec{y})(-1)^{\vec{s}\cdot\vec{a}+\vec{t}\cdot\vec{b}}
=\Tr{U\proj{0}^{\ot n}\ot^m_{i=1}X^{s_i}Z^{t_i} U^\dag \proj{\vec{y}}\ot I_{n+m-k}}\prod^m_{i=1}\left(\frac{\psi^{(i)}_{s_it_i}}{2}\right).
\end{eqnarray}

Now let us define the reference Hermitian operator with respect to $\psi$ as follows
\begin{eqnarray}\label{eq:Oref}
O(\psi):=\frac{1}{2}(I+\mathrm{sgn}(|\psi_{10}|)X+\sgn(|\psi_{01}|)Z+\sgn(|\psi_{11}|)iXZ),
\end{eqnarray}
where the function $\sgn$ is defined as $\sgn(x)=1$ if $x>0$, $\sgn(x)=0$ if $x=0$.
It is easy to verify that $\Tr{O(\psi)}=1, \Tr{O(\psi)^2}=\frac{\chi(\psi)}{2}$, where $\chi(\psi)$ is the Pauli rank of $\psi$.
Besides, we have 
\begin{eqnarray}
|O_{st}|=|\Tr{X^s OZ^t}|=\sgn(|\psi_{st}|). 
\end{eqnarray}
Combined with \eqref{eq:p_p}, we have the following relation
\begin{eqnarray}\label{eq:p_o}
|\psi_{st}|\leq (1-\mu(\psi))^{w(s, t)}|O_{st}|,
\end{eqnarray}
for any $s,t\in\mathbb{F}_2$.
We also define $o_{\vec{a}, \vec{b}}$ as follows
\begin{eqnarray}
o_{\vec{a}, \vec{b}}(\vec{y})=\Tr{U\proj{0}^{\ot n}\ot (X^{\vec{b}}Z^{\vec{a}}\ot^m_{i=1}O_iZ^{\vec{a}} X^{\vec{b}}) U^\dag \proj{\vec{y}}\ot I_{n+m-k}},
\end{eqnarray}
where each $O_i$ is the reference Hermitian operator 
with respect to $\psi_i$ defined as \eqref{eq:Oref}
and the corresponding Fourier coefficient is 
\begin{eqnarray}
\hat{o}_{\vec{s}, \vec{t}}(\vec{y})=
\mathbb{E}_{\vec{a}\in\mathbb{F}^m_2,\vec{b}\in\mathbb{F}^m_2 }o_{\vec{a}, \vec{b}}(\vec{y})(-1)^{\vec{s}\cdot\vec{a}+\vec{t}\cdot\vec{b}}
=\Tr{U\proj{0}^{\ot n}\ot^m_{i=1}X^{s_i}Z^{t_i} U^\dag \proj{\vec{y}}\ot I_{n+m-k}}\prod^m_{i=1}\left(\frac{O^{(i)}_{s_it_i}}{2}\right).
\end{eqnarray}

Thus, in terms of the relation \eqref{eq:p_o}, we have 
\begin{eqnarray}
|\hat{q}_{\vec{s}, \vec{t}}(\vec{y})|
\leq (1-\mu)^{w(\vec{s}, \vec{t})} |\hat{o}_{\vec{s}, \vec{t}}(\vec{y})|,
\end{eqnarray}
where $\mu=\min_i\mu(\psi_i)$ and $w(\vec{s}, \vec{t})$ is defined as \eqref{eq:w}.

Let us define
\begin{equation}
\hat{q}'_{\vec{s},\vec{t}}(\vec{y})=
             \begin{cases}
             \hat{q}_{\vec{s}, \vec{t}}(\vec{y}), & w(\vec{s},\vec{t})\leq l   \\
             0, &  \text{otherwise}
             \end{cases},
\end{equation}
which gives a family of unnormalized probability distribution $\set{q'_{\vec{a},\vec{b}}}$as
$q'_{\vec{a},\vec{b}}(\vec{y})=\sum_{\vec{s}, \vec{t} }\hat{q}'_{\vec{s},\vec{t}}(\vec{y})(-1)^{\vec{s}\cdot \vec{a}+\vec{t}\cdot\vec{b}}$ for each output $\vec{y}\in\mathbb{F}^k_2$
Similar to the proof of Theorem 1,   we show that $q'_{\vec{0}, \vec{0}}(\vec{y})$ gives a good approximation of $q_{\vec{0}, \vec{0}}(\vec{y})$  with respect to $l_1$ norm
 for a large fraction of 
Clifford circuits. 

It is easy to verify that  the equations \eqref{eq:cli_eq} and \eqref{eq:ave_u} still hold, and we can repeat the process of inequality \eqref{eq:For_a} and obtain the following inequality
\begin{eqnarray*}
\mathbb{E}_{\vec{a}\in\mathbb{F}^m_2, \vec{b}\in \mathbb{F}^m_2}\norm{q'_{\vec{a},\vec{b}}- q_{\vec{a},\vec{b}}}^2_1
\leq2^k(1-\mu)^{2l}\sum_{\vec{y}\in\mathbb{F}^k_2}\mathbb{E}_{\vec{a}\in\mathbb{F}^m_2, \vec{b}\in \mathbb{F}^m_2} o^2_{\vec{a}, \vec{b}}(\vec{y}).
\end{eqnarray*}

By the Lemma \ref{lem:2-des}, we have 
\begin{eqnarray*}
\mathbb{E}_{U\sim \mathcal{C}l_{n+m}}o^2_{\vec{a}, \vec{b}}(\vec{y})
\leq 2^{-n-m-k}\prod^m_{i=1}\frac{\chi(\psi_i)}{2}+2^{-2k}.
\end{eqnarray*}
Since $k\leq n+m-\sum^m_{i=1}\log_2\left(\frac{\chi(\psi_i)}{2}\right)$ ,  then we have
\begin{eqnarray*}
\mathbb{E}_{U\sim \mathcal{C}l_{n+m}}\norm{q'_{\vec{0},\vec{0}}- q_{\vec{0},\vec{0}}}^2_1
\leq 2e^{-2\mu l}.
\end{eqnarray*}
By Markov's inequality, we have
\begin{eqnarray*}
\mathrm{Pr}_{U\sim  \mathcal{C}l_{n+m}}\left[\norm{q'_{\vec{0},\vec{0}}- q_{\vec{0},\vec{0}}}_1\leq \sqrt{\alpha}e^{-\mu l}\right]\geq 1-\frac{2}{\alpha}.
\end{eqnarray*}
Therefore, to obtain the $l_1$ norm up to $\delta$, we need take
$l=O(\log(\sqrt{\alpha}/\delta)/\mu)$ and evaluate the Fourier coefficients $\hat{q}'_{\vec{s}, \vec{t}}(\vec{y})$ with $w(\vec{s}, \vec{t})\leq l$,
where the total amount of such Fourier coefficients  is $\sum_{i\leq l}3^i C^i_m\leq 3^lm^l$. Thus, 
 there exists a classical algorithm 
to approximate each output probability $q(\vec{y})$ in time $O((n+m)^3)m^l=(n+m)^{O(1)}m^{O(\log(\sqrt{\alpha}/\delta)/\mu)}$ with $l_1$ norm less than $\delta$ for at least $1-\frac{2}{\alpha}$  fraction of
Clifford circuits.
Thus, we finish the proof of Theorem 3. 

Moreover, if the quantum circuit $\mathcal{C}$ is slightly beyond the Clifford circuits, e.g.
$\mathcal{C}=\mathcal{C}_1\circ V$ where the gates in $\mathcal{C}_1$
are Clifford gates and $V$ is some unitary gate in third level of Clifford  Hierarchy, then the result in 
Theorem 3 still works, as the unitary in third level of Clifford hierarchy maps Pauli operators to Clifford unitaies and thus the
discussion in Appendix \ref{subsec:slig_cl}  still works.

\subsection{Property of Pauli rank}
At the end of this section, let us introduce several basic properties of Pauli rank. For any pure state $\ket{\psi}$ on 
$n$ qubits, we have the Bloch sphere representation 
\begin{eqnarray*}
\proj{\psi}=\frac{1}{2^n}
\sum_{\vec{s},\vec{t}\in \mathbb{F}^n_2}
\psi_{\vec{s},\vec{t}}X^{\vec{s}}Z^{\vec{t}},
\end{eqnarray*}
where $\psi_{\vec{0},\vec{0}}=1$ and 
$\sum_{(\vec{s},\vec{t})\neq (\vec{0},\vec{0})}|\psi_{\vec{s},\vec{t}}|^2=2^n-1$. 
The Pauli rank is defined as the number of nonvanishing coefficients $\psi_{\vec{s},\vec{t}}$,
 that is, 
 \begin{eqnarray}
 \chi(\psi):=|\set{(\vec{s},\vec{t})\in\mathbb{F}^{2n}_2| \psi_{\vec{s}, \vec{t}}\neq 0}|.
 \end{eqnarray}
 Then we have the following property for the Pauli rank.
 
\begin{prop}
Given an $n$-qubit pure state $\ket{\psi}$, we have

(i) $2^n\leq \chi(\psi)\leq 4^n$, $\chi(\psi)=2^n$ iff $\psi$ is a stabilizer state.

(ii) $\chi(\psi_1\ot \psi_2)=\chi(\psi_1)\chi(\psi_2)$.
\end{prop}
\begin{proof}

(i) $2^n\leq \chi(\psi)\leq 4^n$ follows directly  from the definition. 
We only need prove  $\chi(\psi)=2^n$ iff $\psi$ is a stabilizer state. 
In the backward direction, if $\psi$ is a stabilizer state, then it can be written as 
$\proj{\psi}=\prod^n_{i=1}\frac{I+P_i}{2}$, where $P_i\in P^n $  and $P_i$ are commuting with each other. 
Thus, the Pauli rank of $\ket{\psi}$  is $2^n$. In the forward direction, 
if $\chi(\psi)=2^n$, then it can be represented as 
$\proj{\psi}=\frac{1}{2^n}\sum^{2^n}_{i=1} P_i$ where
$P_1=I$,  each $P_i\in P^n$, and $P_i, P_j$ are not equivalent in the sense that $\Tr{P_iP_j}=0$ for any $i\neq j$.
First, we show that $P_iP_j=P_jP_i$ for any $i,j$. Otherwise, there
exists $i_0,j_0$ such that 
$P_{i_0}P_{j_0}=-P_{j_0}P_{i_0}$. 
Since $\psi$ is a pure state, then
\begin{eqnarray}\label{eq:arg2}
\proj{\psi}=\proj{\psi}^2
=\frac{1}{4^n}\sum^{2^n}_{i,j=1}P_iP_j
=\frac{1}{4^n}\sum^{2^n}_{\substack{i,j=1,\\ \set{i,j}\neq \set{i_0, j_0}}}P_iP_j
=\frac{1}{2^n}\sum^{2^n}_{k=1}\frac{n_k}{2^n}P_k,
\end{eqnarray}
where the third inequality comes from the fact that $P_{i_0}P_{j_0}=-P_{j_0}P_{i_0}$. 
Since each $P_iP_j$ is equal to $i^{c_{ijk}}P_k$ for some $k$ and $n_k$ is the summation the these phases $i^{c_{ijk}}$, thus 
\begin{eqnarray}\label{eq:arg}
\sum^{2^n}_{k=1}|n_k|\leq |\set{(i,j)| 1\leq i,j\leq 2^n, \set{i,j}\neq \set{i_0, j_0}}|=4n-2.
\end{eqnarray}
Then 
there is some $k_0$ such that $|n_{k_0}|\leq 2^n-1$, i.e., $\frac{|n_{k_0}|}{2^n}<1$, which contradicts with the representation of 
$\psi$. Thus, $P_i$ are commuting with each other. Next, we prove that  this set of $\set{P_i}^{2^n}_{i=1}$ can be generated by 
some subset $S$  up to $\pm$ sign. 
For any $P_i$ not equal to identity, e.g.,  $P_2$, then there exists 
$U_1\in \mathcal{C}l_n$ such that $U_1P_2U^\dag_1= Z\ot I_{n-1} $,
and  for any $i$, $U_1P_iU^\dag_1$ must have the form 
$Z^{a_i}\ot P_{i, n-1}$, where $P_{i, n-1}\in P^{n-1}$ and they are commuting with each other. 
The generating set $S=\set{ Z\ot I_{n-1}}$. 
For some $P_{i,n-1}$ not equal to identity, e.g., $Z^{a_3}\ot P_{3,n-1 }$, there exists $U_2\in \mathcal{C}l_{n-1}$ such that 
$U_2U_1P_3U^\dag_1U^\dag_2=Z^{a_3}\ot Z\ot I_{n-2}$, and  $U_2U_1P_iU^\dag_1U^\dag_2=Z^{a_i}\ot Z\ot P_{i, n-2}$. Then 
the generating set is updated to
$S=\set{Z\ot I_{n-1}, Z^{a_3}\ot Z\ot I_{n-2} }$. 
Let us repeat the above process for another  $n-2$ times,  finially  we will get some
generating set $S=\set{g_i}^n_{i=1}$, where $g_i=Z^{c_{i,1}}\ot \ldots \ot Z^{c_{i,i-1}} \ot Z\ot I_{n-i}$. 
Moreover, the remaining Pauli operators must have the form 
$\pm \ot^n_{i=1}Z^{a_i}$, which can be generated by  the generating set $S$ up to $\pm$ sign. 
That is, there is a Clifford unitary map $U$ that maps $\proj{\psi}$
to another pure state $\proj{\psi'}=\frac{1}{2^n}\sum_{\vec{a}\in \mathbb{F}^n_2}c_{|\vec{a}|}Z^{\vec{a}}$ where 
$c_{|\vec{a}|}=\pm 1$,  $|\vec{a}|:=\sum_i a_i2^{i-1}$ and $c_0=1$. 
Repeating  the argument  \eqref{eq:arg2} and \eqref{eq:arg} for the pure state $\proj{\psi'}$, we have $c_{|\vec{a}|}=\prod^n_{i=1}c^{a_i}_{2^{i-1}}$. 
Thus $\proj{\psi'}=\prod^n_{i=1}\frac{I+c_{2^{i-1}}Z_i}{2}$ where $Z_i$ denotes the Pauli $Z$ operator acting on the $i$th qubit. 
Therefore
$\psi$ is a stabilizer state.

(ii) This property follows directly from the definition.

\end{proof}

Based on the above proposition and the fact that the Pauli rank is invariant under conjugation by Clifford unitaries, it is easy to see that the
Pauli rank is a good candidate to quantify the magic in a state. Here, using the Pauli rank as a magic monotone is advantageous because it is easier to compute than previous magic monotones \cite{BravyiPRL2016,HowardPRL2017,Veitch2014}, which typically involve a minimization over all stabilizer states.

\end{document}